\documentclass[letterpaper, 10 pt, conference]{ieeeconf}  

\IEEEoverridecommandlockouts                              
\title{\LARGE \bf
	Projected Push-Sum Gradient Descent-Ascent for Convex Optimization with Application to Economic Dispatch Problems}

\author{Jan Zimmermann, Tatiana Tatarenko, Volker Willert, J\"urgen Adamy
\thanks{The authors are with the Control Methods and Robotics Lab at TU Darmstadt, Germany.}
\thanks{The work was gratefully supported by the German Research Foundation
	(Deutsche Forschungsgemeinschaft, DFG) within the SPP 1984 ``Hybrid and multimodal energy systems: System theoretical methods for the transformation and operation of complex networks''.}
}

\usepackage{cite}
\usepackage{amsmath,amssymb,amsfonts}
\usepackage{graphicx}

\usepackage{amsthm}
\usepackage{textcomp}
\usepackage{xcolor}
\usepackage{bm}
\let\labelindent\relax
\usepackage[shortlabels]{enumitem}
\usepackage[utf8]{inputenc}
\usepackage{pgfplots}

\usetikzlibrary{spy}
\theoremstyle{plain}
\newtheorem{thm}{Theorem}
\newtheorem{prop}{Proposition}
\newtheorem{lemma}{Lemma}
\newtheorem{assum}{Assumption}
\newtheorem{remark}{Remark}
\DeclareMathOperator*{\argmin}{arg\,min}

\begin{document}
	\newcommand{\Qt}{\bm Q[t]}
\newcommand{\Qr}{\bm Q[r]}
\newcommand{\1}{\bm 1}
\newcommand{\invdiagPp}{\left(\text{diag}\left[ \bm P^{t+1} \1]  \right] \right)^{-1}}
\newcommand{\diagPt}{\text{diag}\left[ \bm P^t\1  \right] }
\newcommand{\diagPs}{\text{diag}\left[ \bm P^s\1  \right] }
\newcommand{\diag}{\text{diag}}
\newcommand{\nodeset}{\mathcal{V}}
\newcommand{\edgeset}{\mathcal{E}}
\newcommand{\qprod}[1]{\bm \Phi (#1)}
\newcommand{\epsx}[1]{\bm \epsilon_i^{\bm x}[#1]}
\newcommand{\epsmu}[1]{\bm \epsilon_i^{\bm \mu_i}[#1]}
\newcommand{\epsxj}[1]{\bm \epsilon_j^{\bm x}[#1]}
\newcommand{\epsmuj}[1]{\bm \epsilon_j^{\bm \mu_j}[#1]}
\newcommand{\leftnorm}{\left| \left|}
\newcommand{\rightnorm}{\right| \right|}
\newcommand{\La}{\mathcal{L}}
\newcommand{\gradLix}[1]{\nabla_{\bm x} {\mathcal{L}_i(#1[t], \bm \mu_i[t]) }}
\newcommand{\gradLimu}[1]{\nabla_{\bm \mu_i} {\mathcal{L}_i(#1[t], \bm \mu_i[t]) }}
\newcommand{\pimin}{p_{i,\text{min}}}
\newcommand{\pimax}{p_{i,\text{max}}}

\maketitle
\thispagestyle{empty}
\pagestyle{empty}

\begin{abstract}
	We propose a novel algorithm for solving convex, constrained and distributed optimization problems defined on multi-agent-networks, where each agent has exclusive access to a part of the global objective function. The agents are able to exchange information over a directed, weighted communication graph, which can be represented as a column-stochastic matrix.
	The algorithm combines an adjusted push-sum consensus protocol for information diffusion and a gradient descent-ascent on the local cost functions, providing convergence to the optimum of their sum. 
	We provide results on a reformulation of the push-sum into single matrix-updates and prove convergence of the proposed algorithm to an optimal solution, given standard assumptions in distributed optimization.
	The algorithm is applied to a distributed economic dispatch problem, in which the constraints can be expressed in local and global subsets.
\end{abstract}

\section{Introduction}
We consider constrained optimization problems that are distributed over multi-agent-networks. In such scenarios, each agent has a local cost function, only known to the respective agent. The overall goal of the network is to minimize the sum of all local functions, while the exact form of the latter should remain private. This type of  problem is known as a social welfare optimization. Objective variables are often subject to a variety of constraints, depending on the application, that need to be considered in the optimization process. For many of such constrained problems, it can be distinguished between global constraints that effect all agents in the system and local constraints that are only relevant to a single agent. An example application is the distributed economic dispatch problem (DEDP), where each agent represents a generator with a distinct cost function. The goal of DEDPs is to minimize the overall cost for producing power, while matching the demand and keeping the production inside the generator's limits. In such problems, the balancing constraint is globally defined, as it constrains the power production of all generators, but the limits of each generator should remain private and therefore local.  Depending on the cost function choice, the resulting problem is either convex or non-convex.\\
We employ first order gradient methods as the core of the optimization strategy. Currently, a lot of work has been dedicated to optimization methods that use gradient tracking instead of the gradient at a distinct point in time. These methods, published for example in \cite{NedicOS17}, \cite{Pu2018} and \cite{Shi15}, have the advantage that a constant step-size can be used for the gradient update, while first order methods usually require a diminishing steps-size sequence for convergence. However, constraints have not been considered in gradient tracking yet.\\
On the other hand, a couple of publications have already been focused on first order methods that are able to respect constraints. A projection-free method that also uses the push-sum algorithm for spreading information, was published in \cite{tatarenko2019} and further analyzed in \cite{Zimmermann2019}. For this method, the constraints are incorporated into the objective  by a penalization function. One property of this approach is that it considers all constraints to be local, which makes this method applicable to a wide range of optimization problems, including the DEDP. However, next to the step-size sequence, a second sequence for a penalization parameter needs to be determined. Choosing those dependent sequences optimally proved to be non trivial \cite{Zimmermann2019}.\\
One of the first projection-based, distributed gradient methods was published in \cite{NedicOP10}. However, the proposed method relies on double-stochastic matrix-updates, which restrict the communication to undirected graphs. The contributions in \cite{Li2019} and \cite{VanMai2016} rely on row-stochastic communication matrices for diffusing the projected gradient information. Finally, \cite{Tsianos2012} employs the push-sum consensus combined with a projection that uses a convex proximal function. The major drawback of the mentioned projection-based methods in relation to the specific structure of the DEDP under consideration is the assumption that all constraints of the distributed optimization problem are known by every agent and therefore global. This restricts the privacy of the agents with local constraints. Compared to the projection-free method in \cite{tatarenko2019}, they have the advantage that no penalization parameter sequence needs to be chosen.\\
Our work seeks to combine the advantage of the projection methods' reduced parameter number with the ability to respect local constraints, while assuming directed communication architectures. This is done by exploiting the explicit distinction of the constraints into local and global. Similar to the approaches in \cite{Nedic2015} and \cite{tatarenko2019}, we employ the push-sum average consensus for information diffusion. However, by formulating the push-sum algorithm into a single row-stochastic matrix-update for easier analysis, the basic structure of the algorithm is closer to the ones in \cite{Li2019} or \cite{VanMai2016}, which also use row-stochastic updates, but do not rely on the push-sum consensus. For convergence to the optimum, we propose a novel distributed gradient descent-ascent algorithm, which uses a projection in order to incorporate global constraints, while respecting local constraint by adding them over a Lagrange multiplier to the local cost functions.\\
Within this article, we make the following contributions: First, we provide a reformulation of the unconstrained push-sum consensus, which differs to the one published in e.g. \cite{Nedic2015}, and provide convergence properties of the update matrix. Secondly, we prove convergence of the distributed gradient descent-ascent method to an optimal solution of the problem, respecting privacy of the local constraints. At last, we show that our proposed algorithm is applicable to the DEDP. \\
The paper is structured as follows: In section \ref{sec:notationandgraph} we provide our notation for formulas and graphs. The main part begins in section \ref{sec:projpushsumgraddescent} with the formulation of the problem class and the results on the reformulation of the push-sum, before our algorithm for solving the defined problems is proposed. The subsequent section \ref{sec:convergence} provides the convergence proof of the proposed method. In section \ref{sec:simulation} we undergird the theoretic results by a simulation of a basic DEDP, before we summarize and conclude our results in section \ref{sec:conclusion}.

%

\section{Notation and Graphs} \label{sec:notationandgraph}
Throughout the paper we use the following notation: The set of non-negative integers is denoted by $\mathbb{Z}^+$. All time indices $t$ in this work belong to this set. To differentiate between multi-dimensional vectors and scalars we use boldface. $||\cdot||$ denotes the standard euclidean norm. $1, ..., n$ is denoted by $[n]$. The element $ij$ of matrix $\bm M$ is denoted by $\bm M_{ij}$. The operation $\prod_{\mathcal{X}}(u)$ is the projection of $u$ onto the convex set $\mathcal{X}$ such that $\prod_{\mathcal{X}}(u) = \argmin_{x \in \mathcal X} \leftnorm x - u\rightnorm$. Instances of a variable $x$ at time $t$ are denoted by $x[t]$.\\
The directed graph $G = \lbrace \mathcal{V}, \mathcal{E} \rbrace$ consists of a set of vertices $ \mathcal{V}$ and a set of edges  $\mathcal{E} = \mathcal{V} \times \mathcal{V}$. Vertex $j$ can send information to vertex $i$ if $(j,i) \in \mathcal{E}$. The communication channels can be described by the Perron matrix $\bm P$, where $\bm P_{ij} > 0$ if $(j,i) \in \mathcal{E}$ and zero otherwise. This notation includes self-loops such that $\bm P_{ii} > 0$. The set containing the in-neighborhood nodes is  $\mathcal{N}_i^+$, while $\mathcal{N}_i^-$ is the set of out-neighbors. The out degree of node $i$ is denoted with $d_i = |\mathcal{N}_i^-|$. We say that a directed graph is strongly connected if there exists a path from every vertex to every other vertex.

\section{Projected push-sum gradient descent-ascent}\label{sec:projpushsumgraddescent}
\subsection{Problem formulation and matrix update of the push-sum consensus}
We consider optimization problems of the form
\begin{subequations}\label{eq:optproblem}
	\begin{align}
		&\min_{\bm x} F(\bm x) = \min_{\bm x} \sum_{i=1}^n F_i(\bm x), \\
		&\text{s.t. } \bm g_i(\bm x) \leq \bm 0, \text{ }   \bm g_i(\bm x) = (g_{i1}(\bm x), ..., g_{im}(\bm x))^T, \\
		& \quad \ \bm x \in \mathcal{X} \subset \mathbb{R}^d,
	\end{align}
\end{subequations}
where functions $F_i:\mathbb{R}^d \rightarrow \mathbb{R}$ and $\ g_{ij}:\mathbb{R}^d \rightarrow \mathbb{R}$ are differentiable for $i = [n] \text{ and for }  j =  [m]$\footnote{For the sake of notation, we assume here that all agents have $m$ local constraints. However, the following considerations hold for arbitrary, yet finite numbers of local constraints that can differ between the agents.}. While we consider $\bm g_i(\bm x)$ to be local constraint functions of agent $i$, the constraint set $\mathcal{X}$ is assumed to be global and therefore known by every agent.\\
After defining $\mathcal{M}_i^0 = \lbrace \bm \mu_i \in \mathbb{R}^m |\bm \mu_i \succeq 0 \rbrace$ and by using $\bm \mu = (\bm \mu_i^T)_{i=1}^n$, the dual function of the problem takes the form
\begin{equation}
	q(\bm \mu) = \inf_{\bm x \in \mathcal{X}} \left \lbrace \sum_{i=1}^n F_i(\bm x) + \bm \mu_i^T\bm g_i(\bm x) \right \rbrace,
\end{equation}
with which the dual problem
\begin{subequations}
	\begin{align}
		&\max_{\bm \mu} q(\bm \mu), \\
		&\text{s.t. } \bm \mu \in \mathcal{M}^0,
	\end{align}
\end{subequations}

with $\mathcal{M}^0 = \lbrace \bm \mu \in \mathbb{R}^{n \times m}| \bm \mu_i \in \mathcal{M}_i^0, \  i = [n] \rbrace$ can be defined. In accordance to that, the global Lagrangian is the sum of the local Lagrangians
\begin{equation}
	\La (\bm x, \bm \mu) = \sum_{i=1}^n \La_i(\bm x, \bm \mu_i) = \sum_{i=1}^n F_i(\bm x) + \bm \mu_i^T\bm g_i(\bm x) .\\
\end{equation}
Before we continue with the analysis of the push-sum consensus for information diffusion, we make the following assumptions regarding the above problem:
\begin{assum} \label{as:dualitygap}
	$F$ is strongly convex on $\mathcal{X}$ and $\bm g_i$ is convex  for $i = [n]$.
	The optimal value $F^*$ is finite and there is no duality gap, namely $F^* = q^*$. There exist compact sets $\mathcal{M}_i \subset \mathcal{M}_i^0, i = [n]$ containing the dual optimal vectors $\bm \mu_i^*, i = [n]$.
\end{assum}

\begin{assum}\label{as:set}
	$ \mathcal{X} \subset \mathbb{R}^d$ is convex and compact.
\end{assum}

\begin{remark}\label{rem:slater}
	Given the convexity properties of the problem and Slater's constraint qualification, we have strong duality, which implies that the duality gap is zero. Furthermore, the optimal vector for $\bm \mu_i$ is then uniformly bounded in the norm (see \cite{Hiriarturruty1996}). Thereby, Assumption \ref{as:dualitygap} is given, for example, if Slater's condition holds, $F$ and $\bm g_i, i = [n]$, are continuous and the domain $\mathcal{X}$ is compact.
\end{remark}

\begin{assum}\label{as:gradbound}
	The gradients $\nabla_{\bm x} \La_i(\bm x, \bm \mu_i)$, $\nabla_{\bm \mu_i} \La_i(\bm x, \bm \mu_i)$ exist and are uniformly bounded on $\mathcal{X}$ and $\mathcal{M}_i$, i.e. $\exists L_{\bm x} < \infty, L_{\bm \mu_i} < \infty$ such that $||\nabla_{\bm x} \La_i(\bm x, \bm \mu_i)|| \leq L_{\bm x} $   for $\bm x \in \mathcal{X}, \bm \mu_i \in \mathcal{M}_i$ and $||\nabla_{\bm \mu_i} \La_i(\bm x, \bm \mu_i)|| \leq L_{\bm \mu_i} $   for $\bm x \in \mathcal{X}, \bm \mu_i \in \mathcal{M}_i$.
\end{assum}

\begin{remark}
	Note that this means that for either fixed $\bm \mu_i$ or fixed $\bm x$ the Lagrangian $\La_i(\bm x, \bm \mu_i)$ is Lipschitz-continuous with constant $L_{\bm x} $,  $L_{\bm \mu_i}  $,  respectively.
\end{remark}

\begin{assum}\label{as:graph}
	The Graph $\mathcal{G} = \lbrace\mathcal{V}, \mathcal{E}\rbrace$ is fixed and strongly connected. The associated Perron matrix $\bm P$ is column-stochastic.
\end{assum}
\begin{remark}
	If, for example, agent $i$ weights its messages by $1/(d_i)$, with $d_i$ representing the out degree of $i$, the resulting communication matrix contains the elements $\bm P_{ij} = 1/d_i$, which achieves column-stochasticity of $\bm P$. 
\end{remark}

Recall the push-sum consensus protocol from \cite{Charalambous2014}, \cite{kempe2003}
\begin{subequations}\label{eq:consensus}
	\begin{align}
		\bm y[t+1]   & = \bm P \bm y[t],                \\
		\bm x[t+1]   & = \bm P \bm x[t],                \\
		\bm z_i[t+1] & = \frac{\bm x_i[t+1]}{y_i[t+1]},
	\end{align}
\end{subequations}

with initial states $\bm x[0] = \bm z[0] = \bm x_0$ and $\bm y[0] = \1$. The agent-wise update of $\bm z$ can be rewritten as a matrix update, as it is done for example in \cite{Nedic2018}. For that, define the time dependent matrix $\Qt$ such that
\begin{equation} \label{eq:Qt}
	\bm Q(\bm y[t+1], \bm y[t]) = \Qt = \diag(\bm y[t+1])^{-1} \bm P \diag(\bm y[t]).
\end{equation}
Thereby, we can merge the update equations of $\bm x$ and $\bm z$ into
\begin{equation}
	\bm z[t+1] = \Qt \bm z[t]. \label{eq:pushsumzupdate}
\end{equation}
Some important properties of $\Qt$ can now be proven, which hold independently of the values of $\bm y[t]$ at different time instances $t$. Those are summarized in the following Lemma. But first, we introduce the matrix
\begin{equation} \label{eq:Phi}
	\qprod{t,s} = \bm Q[t]\bm Q[t-1] ... \bm Q[s+1] \bm Q[s], \ t > s
\end{equation}
with $\qprod{t,t} = \Qt$, for easier notation.
\begin{lemma} \label{lemma:Q}
	Given Assumption \ref{as:graph},  the time-dependent communication matrix $\Qt$ of equation \eqref{eq:Qt} and the matrix product $\qprod{t,s}$ defined in \eqref{eq:Phi} have the following properties:
	\begin{enumerate}[a)]
		\item Matrix $\Qt$ and matrix  $\qprod{t,s}$ are row-stochastic for $ 0 \leq s \leq t$ and all $t$. \\
		\item  $\lim_{t\rightarrow \infty} \qprod{t,0} = \frac{1}{n} \1\1^T$.  \\
		\item  $ \lim_{t\rightarrow \infty} \qprod{t,s} = \frac{1}{n} \1 \bm y[s]^T $ for finite $ 0 \leq s < t$ .
	\end{enumerate}
\end{lemma}
\begin{proof}
	Part a): \\
	According to equation \eqref{eq:Qt}, we can write
	\begin{align*}
		\qprod{t,s}  = \diag(\bm y[t+1])^{-1} \bm P^{t+1-s}  \diag(\bm y[s]).
	\end{align*}
	For $s=0$ it holds that
	\[\qprod{t,0} \1 = \diag(\bm y[t+1])^{-1} \bm P^{t+1}\bm I \1 = \1,\]
	because, with $\bm y[t+1] = \bm P^{t+1}\1$, each dimension of $\bm y[t+1]$ contains the sum over the respective row of $\bm P^{t+1}$. Therefore, $\diag(\bm y[t+1])^{-1}$ norms the rows of $\bm P^{t+1}$, such that the above holds.
	As this is true for arbitrary $t$, we can factor out $\Qt$ and show row-stochasticity for all $\Qt$:
	\[ \qprod{t,0}\1 = \bm Q[t] \qprod{t-1,0}\1 = \Qt\1 = \1.\]
	Thereby, we also have for $0 \leq s \leq t$
	\[ \qprod{t,s}\1 = \bm Q[t]\bm Q[t-1]...\bm Q[s]\1 = \1.\]
	Part b): \\
	From the Perron-Frobenius Theorem \cite{Seneta1981}, we know that, for a column-stochastic matrix $\bm P$, the limit
	\[\lim_{t\rightarrow \infty} \bm y[t+1] = \lim_{t\rightarrow \infty} \bm P^{t+1} \1 = \bm w \1^T \1  = n \bm w, \]
	holds, with $\bm w$ being the right eigenvector of $\bm P$ for the eigenvalue $\lambda = 1$.
	Therefore,
	\[\lim_{t\rightarrow \infty} \qprod{t,0} = \frac{1}{n} \left(\text{diag}[\bm w ]\right)^{-1} \bm w \1^T = \frac{1}{n} \1 \1^T,\]
	which is a double-stochastic matrix.\\
	Part c): \\
	Using the results from b), we can write for finite $s < t$
	\[\lim_{t\rightarrow \infty} \qprod{t,s} = \frac{1}{n} \1 \1^T\text{diag}(\bm y[s]) = \frac{1}{n} \1 \bm y[s]^T.\]
	Using the column-stochasticity of $\bm P$, we have
	\[\frac{1}{n} \1 \bm y[s]^T \1 = \frac{1}{n} \1 \1^T  (\bm P^s)^T \1 = \frac{1}{n} \1 \1^T \1  = \1,  \]
	which shows row-stochasticity.
\end{proof}
The following Lemma provides us with bounds on the matrix updates.
\begin{lemma} \label{lemma:boundclambda}
	Given Assumption \ref{as:graph}. The matrix $\Qt$ is defined according to \eqref{eq:Qt} and $\qprod{t,s}$ as in \eqref{eq:Phi}. Then, there exist constants $C > 0$ and $\lambda \in (0,1)$  that satisfy the following expressions for $i, j = [n]$, $0 \leq s \leq t$ and $\forall t$:
	\begin{enumerate}[a)]
		\item \begin{equation} \label{eq:lemma2a}
			\left| \qprod{t,0}_{ij} - \frac{1}{n} \sum_{i=1}^{n} \qprod{t,0}_{ij}  \right| \leq C \lambda^t
		\end{equation}
		
		\item \begin{equation} \label{eq:lemma2b}
			\left| \qprod{t,s}_{ij} -\frac{1}{n} \sum_{i=1}^n \qprod{t,s}_{ij} \right| \leq C \lambda^{t-s}
		\end{equation}
	\end{enumerate}
\end{lemma}
\begin{proof}
	Part a):\\ 
	We add $-1/n +1/n$ to the term on the left side of the inequality in \eqref{eq:lemma2a} and apply the triangle inequality, what results in
	\begin{equation*}
		\left| \qprod{t,0}_{ij} - \frac{1}{n} \right| + \left| \frac{1}{n}  \sum_{i=1}^{n} \qprod{t,0}_{ij} - \frac{1}{n}  \right|.
	\end{equation*}
	We already showed convergence of the first term to $1/n$ in Lemma \ref{lemma:Q}. Since the column sum of $\frac{1}{n} \1\1^T$ is equal to 1, the second term also converges. Therefore, we can bound above expression by $C \lambda^{t}$ with $C > 0, \lambda \in (0,1)$, which is a standard procedure for row-stochastic, non-negative matrix multiplications, see for example Proposition 1 in \cite{Nedic2015}.\\
	Part b): \\
	Following the same line of thought as in a), we add $ + \frac{1}{n}  \bm y_j[s]^T -  \frac{1}{n}  \bm y_j[s]^T$ to the left side of the inequality in \eqref{eq:lemma2b} and receive:
	\begin{equation*}
		\left| \qprod{t,s}_{ij} -  \frac{1}{n}  \bm y_j[s]^T \right| + \left |  \frac{1}{n} \sum_{i=1}^n \qprod{t,s}_{ij} -  \frac{1}{n}  \bm y_j[s]^T \right|.
	\end{equation*}
	Again, Lemma \ref{lemma:Q} showed convergence of the first term to zero for finite  $s < t$. Summing again over all columns of the matrix  $\frac{1}{n} \1 \bm y[s]^T$, we receive the vector $\bm y[s]^T$ and therefore convergence of the second term. Note that for $s=t$, the expression does not converge. Thereby, we can bound the  above by $C\lambda^{t-s}$ with $C> 0$ and $\lambda \in (0,1)$, as it is done in the proposition cited in a).
\end{proof}

\subsection{Projected push-sum gradient descent-ascent}
We propose the following agent-wise update equations for solving problem \eqref{eq:optproblem}:
\begin{subequations}\label{eq:algorithm}
	\begin{align}
		&y_i[t+1] = \sum_{j =1}^n \bm P_{ij}y_j[t] , \label{eq:algy} \\
		&\bm z_i[t] = \frac{1}{y_i[t+1] }  \sum_{j =1}^n  \bm P_{ij}y_j[t]  \bm x_j[t], \label{eq:algz}\\
		&\bm x_i[t+1] = \prod_{\mathcal{X}} \Bigg( \bm z_i[t] - \alpha_{t} \frac{\nabla_{\bm x} \La_i(\bm z_i[t], \bm \mu_i[t]) }{y_i[t+1]} \Bigg), \label{eq:algx} \\
		&\bm \mu_i[t+1]    = \prod_{\mathcal{M}_i} \Bigg(\bm \mu_i[t] + \alpha_t \nabla_{\bm \mu_i} \La_i(\bm z_i[t], \bm \mu_i[t]) \Bigg). \label{eq:algmu}
	\end{align}
\end{subequations}
The algorithm above is based on the idea of the push-sum consensus protocol in  \eqref{eq:consensus}, combined with the descent-ascent procedure to update the local optimization variable $\bm x_i$ and dual variable $\bm \mu_i$ for each agent $i\in[n]$. Moreover, note that the dual variables are projected on the local sets $\mathcal M_i$, which are defined in Assumption~\ref{as:dualitygap}. We refer the reader to~\cite{Zhu2016} for possible strategies each agent can use to define its own $\mathcal M_i$ locally.\\
The $\bm z_i$- and $y_i$-update equations can be written in the more concise matrix notation
\begin{align*}
	\bm y[t+1] &= \bm P \bm y[t], \\
	\bm z[t] &= \Qt \bm x[t],
\end{align*}
as $\bm P_{ij}y_j[t]/y_i[t+1]$ represent the elements $ij$ of matrix $\Qt$. Remember that, resulting from Assumption \ref{as:graph}, the Perron matrix $\bm P$ is column-stochastic and  $\bm P_{ij} = 0$ if agent $j$ has no communication link to $i$. The local gradients $ \nabla_{\bm x} \La_i(\bm z_i[t], \bm \mu_i[t])$ of each agent are locally weighted with $y_i[t+1]$. Note that $y_i[t] > 0, \text{ for } i = [n]$ and all $t$,  resulting from its initialization with $y_i[0] = 1, i = [n]$, and the update by a column-stochastic, non-negative matrix.\\
We reformulate above procedure for easier analysis. For that, we define the local disturbance terms 
\begin{align*}
	&\epsx{t}  = \prod_{\mathcal{X}} \Biggl( \bm z_i[t] -  \alpha_{t} \frac{\nabla_{\bm x} \La_i(\bm z_i[t] , \bm \mu_i[t]) }{y_i[t+1]} \Biggr)  - \bm z_i[t] ,\\
	&\epsmu{t} =   \prod_{\mathcal{M}_i}  \Bigg( \bm \mu_i[t] +\alpha_t \nabla_{\bm \mu_i} \La_i \Big(\bm z_i[t], \bm \mu_i[t]\Big)\Bigg) - \bm \mu_i[t]
\end{align*}
and express equations \eqref{eq:algx} and \eqref{eq:algmu} by
\begin{subequations}
	\begin{align}
		\bm x_i[t+1] &= \bm z_i[t] + \epsx{t} \label{eq:reformulationx},\\
		\bm \mu_i[t+1] &= \bm \mu_i[t] + \epsmu{t}.  \label{eq:reformulationmu}
	\end{align}
\end{subequations}

\section{Convergence of proposed algorithm} \label{sec:convergence}
In what follows, we show convergence of the proposed algorithm in \eqref{eq:algorithm} to an optimal primal dual pair of the distributed problem. For that, we first make a standard assumption in distributed optimization regarding the step-size of the distributed gradient descent and ascent:
\begin{assum}\label{as:stepsize}
	The non-increasing, positive step-size sequence $\alpha_t$ has the properties:\\
	a) $ \lim_{t \rightarrow \infty} \alpha_t = 0 $, b) $\sum_{t=0}^{\infty} \alpha_t = \infty$, c) $ \sum_{t=0}^{\infty} \alpha_t^2 < \infty$.
\end{assum}
For example, this assumption holds true for step-sizes of the form $\alpha_t = \frac{c}{t^\gamma}, \forall t \geq 1,$ with $c> 0$ and $\gamma \in (0.5, 1]$.\\
It is possible to bound the norm of the disturbances $\epsx{t}$ and $\epsmu{t}$, defined in the reformulations \eqref{eq:reformulationx} and \eqref{eq:reformulationmu}, using the non-expansive property of the projection operator and the fact that the update $\bm z[t] = \Qt \bm x[t]$ by the row-stochastic matrix $\Qt$ lies inside the convex constraint set $\mathcal{X}$. This is the case, because all $\bm x_i[t]$ are projected onto said constrained set in the previous time-step and every $\bm z_i[t]$ lies inside the convex hull spanned by $\bm x[t]$. Therefore, $\bm z_i[t]$ and $\bm \mu_i[t]$ must lie inside the sets $\mathcal{X}$ and $\mathcal{M}_i$, respectively, in every time step. A similar approach can be found in \cite{VanMai2016}. This allows us to exploit the Assumption \ref{as:gradbound}  on boundness of the Lagrangian gradients as follows

\begin{align}
	||\epsx{t}|| &\leq \left|\left| \alpha_t \frac{\gradLix{\bm z_i}}{y_i[t+1]} \right|\right| \leq \frac{|\alpha_t|L_{\bm x}}{y_i[t+1]} \label{eq:boundepsilonx}, \\
	||\epsmu{t}|| &\leq \left|\left| \alpha_t \gradLimu{\bm z_i}\right|\right| \leq |\alpha_t|L_{\bm \mu_i} \label{eq:boundepsilonmu}.
\end{align}
Using the step-size properties of Assumption \ref{as:stepsize}, it can be concluded that
\begin{align}
	\lim_{t \rightarrow \infty} \alpha_t = 0 &\implies\!\lim_{t \rightarrow \infty} ||\epsx{t}|| = 0, \lim_{t \rightarrow \infty} ||\epsmu{t}|| = 0, \label{eq:alphaepsilonlimit}\\
	\sum_{t=0}^{\infty} \alpha_t^2 < \infty &\implies\!\sum_{t=0}^{\infty} \alpha_t||\epsx{t}|| < \infty,  \sum_{t=0}^{\infty} \alpha_t||\epsmu{t}|| < \infty \label{eq:alphaepsilonsum}
\end{align}
This result will be used in the proof of the following Lemma.
\begin{lemma} \label{lemma:average}
	The Assumptions \ref{as:set}, \ref{as:gradbound} and \ref{as:graph} are given. Denote the average at time $t$ with $\bar{\bm x}[t] = \frac{1}{n} \sum_{i=1}^{n} \bm x_i[t]$. Then,
	\begin{enumerate}[a)]
		\item if Assumption \ref{as:stepsize} a) is true, $	\lim_{t\rightarrow \infty} ||\bm x_i[t] - \bar{\bm x}[t] || = 0.$
		
		\item if Assumption \ref{as:stepsize} c) is true, 	$	\sum_{t=0}^{\infty} \alpha_t ||\bm x_i[t] - \bar{\bm x}[t]|| <\infty.$
	\end{enumerate}
\end{lemma}

\begin{proof}
	Part a):\\
	Expanding equation \eqref{eq:reformulationx}, $\bm x_i[t]$ can be expressed as
	\begin{align*}
		\bm x_i[t] &= \sum_{j=1}^n \qprod{t\!-\!1,0}_{ij} \bm x_j[0]\!+\!\sum_{s=0}^{t-2} \sum_{j=1}^n \qprod{t\!-\!1, s\!+\!1}_{ij} \epsxj{s}\\
		&+\epsxj{t-1}.
	\end{align*}
	Inserting this into $||\bm x_i[t] - 1/n \sum_{i=1}^n \bm x_i[t]||$, repeatedly applying the triangle inequality and using the results from Lemma \ref{lemma:boundclambda}, we receive
	\begin{align*}
		&||\bm x_i[t] - \bar{\bm x}[t]|| \leq   C \lambda^t  \sum_{j=1}^n  \leftnorm \bm x_j[0]  \rightnorm \\
		& + \sum_{s=0}^{t-2}   C_s \lambda_s^{t-s-2}   \sum_{j=1}^n \leftnorm  \epsxj{s} \rightnorm \\
		& + \leftnorm \epsx{t-1} \rightnorm + \frac{1}{n} \sum_{i=1}^n\leftnorm \epsx{t-1}  \rightnorm.
	\end{align*}
	We are now considering the limit $t \rightarrow \infty$ for each line separately.\\
	The expression on the right side of the first line converges to zero, because $\lambda \in (0,1)$ and $|| \bm x_j[0]||, j = [n]$ can assumed to be finite. For the second line, we use Lemma 7 from \cite{NedicOP10}, stating that if for some positive scalar sequence $\gamma_t$ it holds that $\lim_{t \rightarrow \infty} \gamma_t = 0$, then
	\begin{equation*}
		\lim_{t \rightarrow \infty} \sum_{s=0}^{t} \beta^{t-s} \gamma_s = 0,
	\end{equation*}
	where $\beta \in (0,1)$.
	From implication \eqref{eq:alphaepsilonlimit} we know that, given Assumption \ref{as:stepsize}, the limit of the positive, scalar sequence  $\leftnorm  \epsxj{s} \rightnorm$ converges to zero for $j = [n]$. Therefore, we can apply the results of Lemma 7 from \cite{NedicOP10} to the second line after the inequality by substituting $k = t-2$ and conclude
	\begin{equation*}
		\lim_{k \rightarrow \infty} \sum_{s=0}^{k}   C_s \lambda_s^{k-s}  \sum_{j=1}^n \leftnorm  \epsxj{s} \rightnorm = 0.
	\end{equation*}
	
	Following from implication \eqref{eq:alphaepsilonlimit}, the third line converges to zero as well, which concludes the proof of part a).\\
	Part b):\\
	We have
	\begin{align}
		\sum_{t=0}^{\infty}& \alpha_t || \bm x_i[t] - \bar{\bm x}[t]|| \nonumber  \leq \sum_{t=0}^{\infty} \alpha_t a_t + \sum_{t=0}^{\infty} \alpha_t b_t  \nonumber \\
		&+\sum_{t=0}^{\infty} \alpha_t \left( ||\epsx{t-1}|| +  \frac{1}{n} \sum_{i=1}^{n} ||\epsx{t-1}||\right) \label{eq:sumestimate}
	\end{align}
	with sequences
	\begin{equation*}
		a_t =\sum_{j=1}^n \left| \qprod{t-1,0}_{ij} - \frac{1}{n} \sum_{i=1}^n \qprod{t-1,0}_{ij} \right| \leftnorm \bm x_j[0]  \rightnorm
	\end{equation*}
	\begin{align*}
		& \text{and } b_t =  \\
		&\sum_{s=0}^{t-2} \sum_{j=1}^n \left| \qprod{t\!-\!1,s\!+\!1}_{ij} - \frac{1}{n} \sum_{i=1}^n \qprod{t\!-\!1,s\!+\!1}_{ij} \right| \leftnorm \epsxj{s} \rightnorm.
	\end{align*}
	From Lemma \ref{lemma:boundclambda} a) we know that the sequence $a_t$ can be bounded by a sequence $a'_t$ as follows
	\begin{equation*}
		a_t \leq C \lambda^t   \sum_{j=1}^n||\bm x_j[0]|| = a_t'.
	\end{equation*}
	The series $\sum_{t=0}^{\infty} a_t'$ converges, as it is a geometric, convergent series with $0 < \lambda <1$. By direct comparison test it follows that $\sum_{t=0}^{\infty} a_t$  converges as well, because $0 \leq a_t \leq a_t'$.\\
	Resulting from Assumption \ref{as:stepsize}, $\alpha_t$ is a positive, non-increasing sequence. Therefore, there exists a $0 < K<\infty$ such that $\alpha_t \leq K$. Because of that, the first element after the inequality sign of equation \eqref{eq:sumestimate} is summable, because
	\begin{equation}
		\sum_{t=0}^{\infty} \alpha_t a_t \leq K \sum_{t=0}^{\infty} a_t < \infty.
	\end{equation}
	Using Lemma \ref{lemma:boundclambda} b), we receive
	\begin{equation*}
		b_t   \leq \sum_{s=0}^{t-2} 	 C_{s} \lambda_{s}^{t-s-2}   \sum_{j=1}^{n}||\epsxj{s} ||.
	\end{equation*}
	Summing over all $t$ and multiplying with $\alpha_t$, we get
	\begin{align*}
		&\sum_{t=0}^{\infty} \alpha_t b_t  \leq \sum_{t=0}^{\infty} \left(\sum_{s=0}^{t-2} C_{s} \lambda_{s}^{t-s-2}   \sum_{j=1}^{n}\alpha_s ||\epsxj{s} ||\right),
	\end{align*}
	where we used the non-increasing property $\alpha_s \leq \alpha_t$ for $s \leq t$. Now, define
	\begin{equation*}
		\gamma_s = \sum_{j=1}^{n} \alpha_s ||\epsxj{s} ||.
	\end{equation*}
	We know from implication \eqref{eq:alphaepsilonsum} that $\sum_{t=1}^{\infty} \gamma_t < \infty$.
	According to Lemma 7 from \cite{NedicOP10}, we know
	\begin{equation*}
		\sum_{t=0}^{\infty} \sum_{s=0}^{t-2} \lambda^{t-s-2} \gamma_s  < \infty.
	\end{equation*}
	Applying this to our case, we conclude
	\begin{equation*}
		\sum_{t=0}^\infty \alpha_t b_t < \infty.
	\end{equation*}
	Finally, we have
	\begin{equation*}
		\sum_{t=0}^{\infty} \alpha_t ||\epsx{t-1}|| \leq \sum_{t=0}^{\infty} \alpha_{t-1} ||\epsx{t-1}|| < \infty,
	\end{equation*}
	where we used again the implication \eqref{eq:alphaepsilonsum} and the fact that $\alpha_t$ is non-increasing.\\
	Therefore,
	\begin{equation*}
		\sum_{t=0}^{\infty}\alpha_t || \bm x_i[t] - \bar{\bm x}[t]|| < \infty
	\end{equation*}
	holds, which concludes the proof.
	
\end{proof}
The above Lemma is necessary for the convergence proof of the suggested algorithm. Next, we provide an upper bound for the algorithm updates in each time step.
\begin{prop}
	Let Assumptions \ref{as:set}, \ref{as:gradbound} and \ref{as:graph} hold. Then, for the optimal primal dual pair  $\bm x^* \in \mathcal{X}$,  $\bm \mu_i^* \in \mathcal{M}_i$ and $t\geq t_0$, the following bound holds
	\begin{align*}
		&\sum_{i=1}^n  \left( y_i[t+1] || \bm x_i[t+1] - \bm x^*||^2 + ||\bm \mu_i[t+1] - \bm \mu_i^*||^2 \right)\\
		& \leq \sum_{i=1}^n \left(  y_i[t] || \bm x_i[t] - \bm x^* ||^2  + ||\bm \mu_i[t] - \bm \mu_i^*||^2 \right) \\
		&- 2 \alpha_t ( \La(\bm {\bar{x}}[t], \bm \mu^*)\!-\!\La(\bm x^*, \bm \mu^*)\!+\!\La(\bm x^*, \bm \mu^*)\!-\!\La(\bm x^*, \bm \mu[t]))\\
		& + 2 \alpha_t L_{\bm x} n  \sum_{i=1}^n || \bm x_i[t]\!-\!\bar{\bm x}[t]|| + L_{\bm x}^2 \alpha_t^2 \sum_{i=1}^n \frac{1}{w_i} + \alpha_t^2 \sum_{i=1}^n L_{\bm \mu_i}^2,
	\end{align*}
	where $\bm \mu^* = (\bm \mu_1^*,\ldots, \bm \mu_n^*)$ and $w_i$ is such that $\lim_{t \rightarrow \infty}  y_i[t] =  w_i$.
\end{prop}

\begin{proof}
	Inserting $\bm x_i[t+1]$ and $\bm{\mu}_i[t+1]$ of equations \eqref{eq:algx} and \eqref{eq:algmu}, using the non-expansive property of the projection operator, the fact that the optimal values $(\bm x^*, \bm \mu^*)$ lie within the sets $\mathcal{X}$ and $\mathcal{M}_i$, respectively, and by expanding the quadratic norm, the following inequality holds
	\begin{subequations}
		\begin{align}
			& y_i[t+1] || \bm x_i[t+1] - \bm x^*||^2 + ||\bm \mu_i[t+1] - \bm \mu_i^*||^2 \label{eq:lemboundright} \\
			&\leq    y_i[t+1] || \bm z_i[t] - \bm x^*||^2 + ||\bm \mu_i[t] - \bm \mu_i^*||^2  \label{eq:lemboundquadnorm}  \\
			& -    2 \alpha_t \gradLix{\bm z_i}^T(\bm z_i[t] - \bm x^*) \label{eq:lemboundgradx}  \\
			& + 2\alpha_t \gradLimu{\bm z_i}^T(\bm \mu_i[t] - \bm \mu_i^*) \label{eq:lemboundgradmu}  \\
			&+ \frac{\alpha_t^2}{y_i[t+1]} ||\gradLix{\bm z_i}||^2 \label{eq:lemboundquadgrad} \\
			&+ \alpha_t^2 || \gradLimu{\bm z_i}||^2, \label{eq:lemboundquadgrad2}
		\end{align}
	\end{subequations}
	Next, we sum the left and right side of above inequality from $i=1$ to $n$ and analyze every line after the inequality sign separately.\\
	Using Jensen's inequality and the fact that $\bm z[t] = \Qt \bm x[t]$, we get  
	\begin{equation*}
		|| \sum_{j=1}^{n} \Qt_{ij} \bm x_j[t] - \bm x^*||^2 \leq  \sum_{j=1}^{n} \Qt_{ij} || \bm x_j[t] - \bm x^*||^2
	\end{equation*}
	Thus, we can bound the first addend in \eqref{eq:lemboundquadnorm} as
	\begin{align*}
		&	\sum_{i=1}^{n} y_i[t+1] \sum_{j =1}^n \frac{\bm P_{ij}y_j[t]}{y_i[t+1]} ||\bm x_j[t] - \bm x^*||^2\\
		&\leq \sum_{j =1}^n y_j[t]||\bm x_j[t] - \bm x^*||^2 \sum_{i=1}^n \bm P_{ij} = \sum_{i =1}^n y_i[t]||\bm x_i[t] - \bm x^*||^2,
	\end{align*}
	where we replaced $\Qt_{ij}$ with its elements in the first line, rearranged the sums in the second and used the column-stochasticity property of $\bm P$.
	With that,
	\begin{align*}
		\sum_{i =1}^n\eqref{eq:lemboundquadnorm} \leq \sum_{i=1}^n \left(  y_i[t]|| \bm x_i[t] - \bm x^*||^2 +||\bm \mu_i[t] - \bm \mu_i^*||^2   \right).
	\end{align*}
	For \eqref{eq:lemboundgradx}, because of convexity of $\La_i( \bm z_i[t], \bm{\mu}_i[t])$ for fixed $\bm{\mu}_i[t]$, we have
	\begin{align*}
		- \gradLix{\bm z_i}^T&(\bm z_i[t] - \bm x^*) \\
		&\leq \La_i(\bm x^*, \bm \mu_i[t])  - \La_i(\bm z_i[t], \bm \mu_i[t]).
	\end{align*}
	In \eqref{eq:lemboundgradmu}, $\La_i( \bm z_i[t], \bm{\mu}_i[t])$ depends affinely on $\bm{\mu}_i[t]$ with fixed $ \bm z_i[t]$ and therefore
	\begin{align*}
		\gradLimu{\bm z_i}^T&(\bm \mu_i[t] - \bm \mu_i^*) \\
		&= \La_i(\bm z_i[t], \bm \mu_i[t]) - \La_i(\bm z_i[t], \bm  \mu_i^*).
	\end{align*}
	Combing above results, adding $ +  \La_i(\bm x^*, \bm \mu_i^*) - \La_i(\bm x^*, \bm \mu_i^*)$ and $ +  \La_i(\bm{\bar x}[t], \bm \mu_i^*) - \La_i(\bm{\bar x}[t], \bm \mu_i^*)$, as well as summing from $i=0$ to $n$, we receive
	\begin{align*}
		\sum_{i =1}^n \eqref{eq:lemboundgradx} + \eqref{eq:lemboundgradmu}& \leq - 2 \alpha_t \big( \La(\bm{\bar x}[t], \bm \mu^*)  - \La(\bm x^*, \bm \mu^*) \\
		&+ \La(\bm x^*, \bm \mu^*)  - \La(\bm x^*, \bm \mu[t])    \big)  \\	
		&- 2\alpha_t\sum_{i=1}^n \left( \La_i(\bm z_i[t], \bm  \mu_i^*) - \La_i(\bm{\bar x}[t], \bm \mu_i^*) \right).\\
	\end{align*}
	The last line in above expression can be further bounded
	\begin{align*}
		- 2\alpha_t &\sum_{i=1}^n \left( \La_i(\bm z_i[t], \bm  \mu_i^*) - \La_i(\bm{\bar x}[t], \bm \mu_i^*) \right) \\
		&\leq 2 \alpha_t\sum_{i=1}^n \left| \La_i(\bm z_i[t], \bm  \mu_i^*) - \La_i(\bm{\bar x}[t], \bm \mu_i^*) \right| \\
		&\leq 2 \alpha_t L_{\bm x} \sum_{i=1}^n || \bm z_i[t] - \bm{\bar x}[t] || \\
		& \leq 2 \alpha_t L_{\bm x} \sum_{i=1}^n \sum_{j=1}^{n}\Qt_{ij} \leftnorm \bm x_j[t] - \bar{\bm x}[t] \rightnorm\\
		& \leq 2 \alpha_t  L_{\bm x} n \sum_{i=1}^n \leftnorm \bm x_i[t] - \bar{\bm x}[t] \rightnorm,
	\end{align*}
	using $L_{\bm x}-$ Lipschitz continuity of the Lagrangian for fixed $\bm \mu_i^*$,  triangle inequality and the fact that $0 \leq \Qt_{ij} < 1$. With that,
	\begin{align*}
		& \sum_{i =1}^n \eqref{eq:lemboundgradx} + \eqref{eq:lemboundgradmu} \leq - 2 \alpha_t \big( \La(\bm{\bar x}[t], \bm \mu^*)  - \La(\bm x^*, \bm \mu^*) \\
		&+ \La(\bm x^*, \bm \mu^*)  - \La(\bm x^*, \bm \mu[t])    \big)
		+ 2 \alpha_t L_{\bm x} n \sum_{i=1}^n \leftnorm \bm x_i[t] - \bar{\bm x}[t] \rightnorm.
	\end{align*}
	There exists a $t_0$ such that for all  $t > t_0$, it holds that $ y_i[t] \geq \frac{nw_i}{2} > \frac{w_i}{2} $ as $\lim_{t \rightarrow \infty} y_i[t] = w_i$.
	Therefore, using the gradient bounds it holds for $t>t_0$ that
	\begin{align*}
		\sum_{i =1}^n \eqref{eq:lemboundquadgrad} + \eqref{eq:lemboundquadgrad2} \leq  L_{\bm x}^2 \alpha_t^2 \sum_{i=1}^n \frac{1}{w_i} + \alpha_t^2 \sum_{i=1}^n L_{\bm \mu_i} ^2.
	\end{align*}
	Combing above results concludes the proof.

\end{proof}
Before we are able to finally prove convergence of our method to the optimum of problem \eqref{eq:optproblem}, we provide the following Lemma, which is the deterministic version of a Theorem in \cite{Robbins1971}:
\begin{lemma}\label{lemma:sequences}
	Let $\lbrace v_t \rbrace_{t=0}^\infty$, $\lbrace u_t \rbrace_{t=0}^\infty$, $\lbrace b_t \rbrace_{t=0}^\infty$ and $\lbrace c_t \rbrace_{t=0}^\infty$ be non-negative sequences such that $\sum_{t=0}^{\infty} b_t < \infty$ and $\sum_{t=0}^{\infty} c_t < \infty$ and
	\begin{equation*}
		v_{t+1} \leq (1 + b_t)v_t - u_t + c_t, \forall t \geq 0.
	\end{equation*}
	Then $v_t$ converges and $\sum_{t=0}^\infty u_t < \infty$.
\end{lemma}
This Lemma will be the key element for proving the following main result:

\begin{thm}
	Let Assumptions \ref{as:dualitygap}-\ref{as:stepsize} hold. Then,
	$\bm x_i[t]$ and $\bm \mu_i[t] $, updated by the rules in \eqref{eq:algy} - \eqref{eq:algmu}, converge to an optimal primal dual pair $(\bm{{x}}^*, \bm{{\mu}}^*_i ) \in \mathcal{X}\times \mathcal{M}$ for $i = [n]$ as $t \rightarrow \infty$.
	
\end{thm}

\begin{proof}
	To apply Lemma \ref{lemma:sequences} let us define
	\begin{align*}
		v_t &=  \sum_{i=1}^n \left(  y_i[t] || \bm x_i[t] - \bm x^* ||^2  + ||\bm \mu_i[t] - \bm \mu_i^*||^2 \right),  \\
		u_t &= 2 \alpha_t ( \La(\bm {\bar{x}}[t], \bm \mu^*) - \La(\bm x^*, \bm \mu^*)  + \La(\bm x^*, \bm \mu^*) \\
		& - \La(\bm x^*, \bm \mu[t])), \\
		c_t &=  2 \alpha_t L_{\bm{x}} n  \sum_{i=1}^n || \bm x_i[t]-\bar{\bm x}[t]|| + L_{\bm x}^2 \alpha_t^2 \sum_{i=1}^n \frac{1}{w_i} \\ 
		&+\alpha_t^2 \sum_{i=1}^n L_{\bm \mu_i},\\
		b_t  &= 0.
	\end{align*}
	For showing that $c_t$ is sumable, we recall from Lemma \ref{lemma:average}b) that, under given assumptions, $\sum_{t=0}^\infty \alpha_t || \bm x_i[t] - \bar{\bm x}[t]|| < \infty$. Therefore,
	\begin{equation*}
		2 L_x n  \sum_{t=0}^\infty  \alpha_t \sum_{i=1}^n ||\bm x_i[t] - \bar{\bm x}[t]|| < \infty.
	\end{equation*}
	By Assumption \ref{as:stepsize}, we directly have $ \left( L_{\bm x}^2  \sum_{i=1}^n \left( \frac{1}{w_i}  + L_{\bm \mu_i} \right) \right)  \sum_{t=0}^\infty \alpha_t^2 < \infty$. Together, this results in 
	\begin{equation*}
		\sum_{t=0}^\infty c_t < \infty.
	\end{equation*}
	Applying Lemma \ref{lemma:sequences}, we can then make the statements
	\begin{align}
		\exists \delta, \lim_{t \rightarrow \infty} v_t &=  \lim_{t \rightarrow \infty} \sum_{i=1}^n \left(  y_i[t] || \bm x_i[t]\!-\!\bm x^* ||^2  + ||\bm \mu_i[t]\!-\!\bm \mu_i^*||^2 \right)  \nonumber \\
		&= \delta \geq 0, \label{eq:limvt}\\
		\sum_{t=0}^\infty  u_t &=  \sum_{t=0}^\infty  2 \alpha_t ( \La(\bm {\bar{x}}[t], \bm \mu^*) - \La(\bm x^*, \bm \mu^*) \nonumber \\
		&+ \La(\bm x^*, \bm \mu^*) - \La(\bm x^*, \bm \mu[t])  < \infty. \label{eq:sumut}
	\end{align}
	Because the sum of the step-size is infinite, $\sum_{t=0}^{\infty} \alpha_t  = \infty$, by assumption, there need to exist subsequences $\bm x[t_l], \bm \mu_i[t_l]$, such that
	\begin{align*}
		\lim_{l \rightarrow \infty}\!\La(\bm {\bar{x}}[t_l], \bm \mu^*)\!-\!\La(\bm x^*, \bm \mu^*)\!+\!\La(\bm x^*, \bm \mu^*)\!-\!\La(\bm x^*, \bm \mu[t_l]) = 0
	\end{align*}
	Because $\La(\bm {x}[t], \bm \mu[t])$ is affine for fixed $\bm {x}[t] = \bm x^*$ and convex for fixed $\bm {\mu}_i[t]= \bm \mu_i^*$, it holds  that, $\forall t_l$,
	\begin{align*}
		\La(\bm {\bar{x}}[t_l], \bm \mu^*)\!-\!\La(\bm x^*, \bm \mu^*) \geq 0,\La(\bm x^*, \bm \mu^*)\!-\!\La(\bm x^*, \bm \mu[t_l])\geq 0.
	\end{align*}
	Therefore, the limit holds, if and only if
	\begin{align*}
		\lim_{l \rightarrow \infty} & \La(\bm {\bar{x}}[t_l], \bm \mu^*) \!=\!\La(\bm x^*, \bm \mu^*), \\
			\lim_{l \rightarrow \infty} &\La(\bm x^*, \bm \mu[t_l])\!=\!\La(\bm x^*, \bm \mu^*).
	\end{align*}
	Following from convergence of $v_t$ to some constant $\delta$, the subsequences $\bm {\bar{x}}[t_l]$ and $\bm \mu[t_l]$ are bounded. With that, we can choose convergent subsequences $\bm {\bar{x}}[t_{l_s}]$ and $\bm \mu[t_{l_s}]$, such that $\lim_{s \rightarrow \infty} (\bm {\bar{x}}[t_{l_s}], \bm \mu[t_{l_s}]) = (\bm{\hat{x}}, \bm{\hat{\mu}}) \in \mathcal{X} \times \mathcal{M}$, since $\mathcal{X}$ and $\mathcal{M}$ are closed. Therefore, it holds that
	\begin{align*}
		\lim_{s \rightarrow \infty} \La(\bm {\bar{x}}[t_{l_s}], \bm \mu^*)  &= \La(\bm{\hat{x}}, \bm \mu^*)  = \La(\bm x^*, \bm \mu^*) \text{ and } \\
		\lim_{s \rightarrow \infty} \La(\bm x^*, \bm \mu[t_{l_s}])  &=  \La(\bm x^*,  \bm{\hat{\mu}})  = \La(\bm x^*, \bm \mu^*).
	\end{align*}
	Resulting from the strong convexity of $F(\bm x)$ over $\mathcal{X}$, the equality $\La(\bm{\hat{x}}, \bm \mu^*)  = \La(\bm x^*, \bm \mu^*) = \min_{\bm x} \La(\bm x, \bm \mu^*)$ implies that $\bm{\hat{x}} = \bm{x}^*$. Due to dual feasibility of $\bm{\hat{\mu}}$,  $\bm{\hat{x}} = \bm{x}^*$ and $\La(\bm x^*,  \bm{\hat{\mu}})  = \La(\bm x^*, \bm \mu^*) = \max_{\bm \mu \geq 0} \La(\bm x^*, \bm \mu)$, it is implied that $(\bm{\hat{x}}, \bm{\hat{\mu}}) = (\bm x^*, \bm \mu^*)$. Next, taking into account Lemma~\ref{lemma:average}a) and  \eqref{eq:limvt}, we obtain
	\begin{align*}
		\delta &=  \lim_{t \rightarrow \infty} \sum_{i=1}^n \left(  y_i[t] || \bm x_i[t] - \bm x^* ||^2  + ||\bm \mu_i[t] - \bm \mu_i^*||^2 \right)  \cr
		& =\lim_{s \rightarrow \infty} \sum_{i=1}^n \left(  y_i[t_{l_s}] || \bar{\bm x}[t_{l_s}] - \bm x^* ||^2  + ||\bm \mu_i[t_{l_s}] - \bm \mu_i^*||^2 \right) \cr
		& = 0.
	\end{align*}
	Finally, as $y_i[t]>0$ for all $t$, we conclude
	\begin{align*}
		&\lim_{t\to\infty}|| \bm x_i[t] - \bm x^* || =0, \cr
		&\lim_{t\to\infty}|| \bm \mu_i[t] - \bm \mu_i^* || =0 \,\mbox{ for $i=[n]$}.
	\end{align*}


\end{proof}

\section{Simulation}\label{sec:simulation}
As motivated in the introduction, we consider an economic dispatch problem as an example application. In this problem, a group of networked generators seeks to fulfill some predefined demand $D$ while minimizing their summed up local cost functions, which are assumed to take a quadratic form. Formally, the problem can be defined by
\begin{subequations}\label{prob:edp}
	\begin{align}
		\min_{\bm p}& \sum_{i=1}^n C_i(p_i) = \min_{\bm p} \sum_{i=1}^n a_i p_i^2 + b_i p_i + c_i, \label{eq:generatorcost} \\
		\text{s. t. } &\sum_{i=1}^n p_i = D, \label{eq:powerbalance} \\
		& p_{i, \text{min}} \leq p_i \leq  p_{i,\text{max}}, \ \forall i \in [n] \label{eq:generatorlimits}.
	\end{align}
\end{subequations}
The power balance constraint \eqref{eq:powerbalance} consists of the sum of all local decision variables $p_i$ and is therefore part of the global constraint set. Furthermore, we add the technical constraint that all power outputs of the generators should be positive $\bm p \geq \bm 0$, defining the global constraint set as $\mathcal{P} = \lbrace \bm p | \bm p \geq \bm 0, \sum_{i=1}^n p_i = D \rbrace$. Thereby, we ensure that $\mathcal{P}$ is closed and bounded and therefore compact, satisfying Assumption \ref{as:set}. Furthermore, this provides that the projection onto   $\mathcal P$  binds the generation $p_i$, because, regardless of other $p_j, j =  [n] \neq i$, it holds that $0 \leq p_i \leq D$.  The lower and upper power limits of each generator in constraint \eqref{eq:generatorlimits} allocate the local part of the constraint set and are therefore exclusively known by the respective agent $i$. \\
The solution set is non-empty if
\begin{equation*}
	\sum_{i=1}^n p_{\text{min},i} \leq D \leq 	\sum_{i=1}^n p_{\text{max},i}.
\end{equation*}
Therefore, there exists at least one relative interior point for which the affine equality and inequality constraints are fulfilled, which satisfies the Slater constraint qualification. Together with the strong convexity of the cost functions, we conclude that Assumption \ref{as:dualitygap} is given for this problem.\\
The local Lagrangians of the problem takes the form
\begin{align*}
	\La(p_i, \bm \mu_i) &= 	a_i p_i^2 + b_i p_i + c_i  \\
	&+ \mu_{i,1}(\pimin - p_i) + \mu_{i,2}(p_i - \pimax)
\end{align*}
To check, whether Assumption \ref{as:gradbound} is satisfied, the gradients of the local Lagrangians are inspected. First, we realize that $\bm p_i[t] \equiv \bm z_i[t]$ of equation \eqref{eq:algz} is uniformly bounded in every time step. This results from projecting the gradient update onto the convex and compact set $\mathcal{P}$ and communicating the results ($\bm x[t+1]$ in equation \eqref{eq:algx}) in the next time step via a row-stochastic communication matrix, such that the result lies inside the convex hull spanned by $\bm x[t+1]$ and therefore lies inside of $\mathcal{P}$. Using this result and the fact that Slater's constraint qualification is given, which provides us with uniform bounds on  $\bm \mu_i$ (see Remark \ref{rem:slater}), we can conclude that both  $\nabla_{\bm p}\La(p_i, \bm \mu_i)$ and  $\nabla_{\bm \mu_i}\La(p_i, \bm \mu_i)$ can be uniformly bounded. Together with the strong convexity of the cost functions \eqref{eq:generatorcost}, we conclude that Assumption \ref{as:gradbound} is given for Problem \eqref{prob:edp}.\\
We chose a simple setup of four generators and designed the demand $D$ and generator limits such that above equation is true. Each agent $i$ maintains an estimation $\bm p_i$, containing all decision variables, i.e. $\bm p_i = (p_i)_{i=1}^n$. The agents were connected by a static, strongly connected graph, such that Assumption \ref{as:graph} is satisfied.
At last, the step size sequence was chosen by a grid search according to Assumption \ref{as:stepsize} with $a_t = 15/t^{0.60}$. \\
\begin{figure}[]
	\centering
	\scalebox{0.9}{
		\input{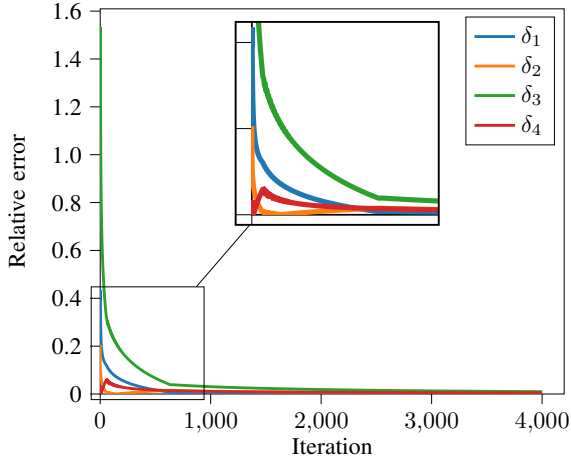} }
	\caption{Convergence of relative error with step size $a_t = 15/t^{0.60}$. Largest error after 4000 iterations: $0.01$. }
	\label{fig:error}
\end{figure}
In Figure \ref{fig:error}, the convergence of the relative errors $\delta_k = |p_k - p_k^*|/p_k^*$, $k = [4]$ of one example agent are depicted. Three of the four states approach 0 already after 500 iterations, while the error $\delta_3$ shows slow convergence over several hundred iterations. After 1000 iterations, all errors are below $0.04$ and after 4000 iterations the highest relative error in the agent system is lower than $0.01$.

\section{Conclusion} \label{sec:conclusion}
In the work at hand, we tailored a solution method to a class of distributed, convex optimization problems that need to respect both global and local constraints. Each agent projects its gradient update onto the global set while updating a Lagrange parameter, over which their local constraints are added to the cost function. Convergence to the optimal value  was proven and some convergence properties shortly discussed by an example of the economic dispatch problem. Future work will include augmenting the method for time-varying communication architectures in order to make the algorithm more robust against failing communication channels.

%

\bibliography{cdcbib}
\bibliographystyle{plain}

\end{document}